%% file: main.tex
\documentclass[12pt]{article}

\input{mary_macros}

\title{Improved Trade-offs Between Amortization and Download Bandwidth for Linear HSS}
\author{Keller Blackwell and Mary Wootters}
\date{}

\begin{document}

\maketitle

\begin{abstract}
    \input{abstract}

\end{abstract}

\section{Introduction}
\input{introduction}

\section{Preliminaries}\label{sec: prelim}
\input{preliminaries}

\section{Equivalence of Linear HSS and Labelweight Codes}\label{sec: equiv}
\input{equivalence}

\section{Linear HSS from Hermitian Codes}\label{sec: hermitian}
\input{hermitian}

\section{Linear HSS from Goppa Codes}\label{sec: goppa}
\input{goppa}

\bibliographystyle{alpha}
\bibliography{refs.bib}

\appendix

\section{Linear HSS from Random Codes}
\label{apx: linear hss from random codes}
\input{gv_generalization}

\end{document}

%% file: mary_macros.tex
\newcommand{\Share}{\mathsf{Share}}
\newcommand{\Eval}{\mathsf{Eval}}
\newcommand{\Rec}{\mathsf{Rec}}
\newcommand{\POLY}{\mathrm{POLY}}
\newcommand{\POLYdmF}{\POLY_{d,m}(\F)}

\usepackage[margin=1.125in]{geometry}
\usepackage{xcolor}
\definecolor{DarkGreen}{rgb}{0.1,0.5,0.1}
\definecolor{DarkRed}{rgb}{0.5,0.1,0.1}
\definecolor{DarkBlue}{rgb}{0.1,0.1,0.5}

\usepackage[small]{caption}
\usepackage[pdftex]{hyperref}
\hypersetup{
    unicode=false,          
    pdftoolbar=true,        
    pdfmenubar=true,        
    pdffitwindow=false,      
    pdfnewwindow=true,      
    colorlinks=true,       
    linkcolor=DarkBlue,          
    citecolor=DarkGreen,        
    filecolor=DarkGreen,      
    urlcolor=DarkBlue,          
    pdftitle={},
    pdfauthor={},    
}
\usepackage{amssymb,amsmath,amsthm,amsfonts,enumitem}
\usepackage{blkarray}
\usepackage{fullpage,nicefrac,comment}
\usepackage{tikz}
\usetikzlibrary{arrows,decorations.pathmorphing,decorations.shapes,patterns,positioning}
\usepackage{multirow}
\usepackage{stmaryrd}
\usepackage[ruled,linesnumbered]{algorithm2e}

\newcommand{\cC}{\ensuremath{\mathcal{C}}}

\newcommand{\cL}{\ensuremath{\mathcal{L}}}
\newcommand{\cM}{\ensuremath{\mathcal{M}}}

\newcommand{\F}{{\mathbb F}}

\newcommand{\fm}{{\mathfrak{m}}}

\newcommand{\inabs}[1]{\left|#1\right|}
\newcommand{\abs}[1]{\inabs{#1}}

\newcommand{\inset}[1]{\left\{#1\right\}}

\newcommand{\ev}{\mathrm{ev}}

\renewcommand{\epsilon}{\varepsilon}

\newtheorem{theorem}{Theorem} 
\newtheorem{lemma}[theorem]{Lemma} 
\newtheorem{definition}{Definition}

\newtheorem{observation}[theorem]{Observation}

\newtheorem{remark}{Remark}
\newtheorem{assume}{Assumption}
\newtheorem{claim}[theorem]{Claim}

\newtheorem{example}{Example}

%% file: abstract.tex
A \emph{Homomorphic Secret Sharing} (HSS) scheme is a secret-sharing scheme that shares a secret $x$ among $s$ servers, and additionally allows an output client to reconstruct some function $f(x)$ using information that can be locally computed by each server.  A key parameter in HSS schemes is \emph{download rate}, which quantifies how much information the output client needs to download from the servers. 
Often, download rate is improved by \emph{amortizing} over $\ell$ instances of the problem, making $\ell$ also a key parameter of interest. 

Recent work \cite{FIKW22} established a limit on the download rate of linear HSS schemes for computing low-degree polynomials and constructed schemes that achieve this optimal download rate; their schemes required amortization over $\ell = \Omega(s \log(s))$ instances of the problem. Subsequent work \cite{BW23} completely characterized linear HSS schemes that achieve optimal download rate in terms of a coding-theoretic notion termed \emph{optimal labelweight codes}. A consequence of this characterization was that $\ell = \Omega(s \log(s))$ is in fact necessary to achieve optimal download rate.

In this paper, we characterize \textit{all} linear HSS schemes, showing that schemes of any download rate are equivalent to a generalization of optimal labelweight codes. This equivalence is constructive and provides a way to obtain an explicit linear HSS scheme from \textit{any} linear code. Using this characterization, we present explicit linear HSS schemes with slightly sub-optimal rate but with much improved amortization $\ell = O(s)$. Our constructions are based on algebraic geometry codes (specifically Hermitian codes and Goppa codes).


%% file: introduction.tex
A \emph{Homomorphic Secret Sharing} (HSS) scheme is a secret sharing scheme that supports computation on top of the shares~\cite{C:Benaloh86a, BGI16, BGILT18}.
Homomorphic Secret Sharing has been a useful primitive in cryptography, with applications ranging from private information retrieval to secure multiparty computation (see, e.g., \cite{BCGIO17, BGILT18}).

In this work, we focus on information-theoretically secure HSS schemes for the class of degree $d$, $m$-variate polynomials. 
  Suppose $m$ secrets, $x_1, \ldots, x_m$, are shared independently with a $t$-private secret sharing scheme $\Share$\footnote{
A $t$-private secret sharing scheme shares a secret $x$ among $s$ servers by computing $s$ \emph{shares}, $\Share(x) = (y_1, \ldots, y_s)$.
The $t$-privacy guarantee means that no $t$ of the servers should be able to learn anything about $x$ given their shares.  
  }%
and that server $j$ receives the $m$ shares $y_{k,j}$ for $k \in [m]$, $j \in [s]$. Denote by $\POLYdmF \subseteq \mathbb{F}[X_1, \ldots, X_m]$ an arbitrary set of degree $d$, $m$-variate polynomials.  Given a polynomial $f \in \POLYdmF$, each server $j$ does some local computation on its shares $y_{k,j}$ for $k \in [m]$ to obtain an \emph{output share} $z_j = \Eval(f,j,(y_{1,j}, \ldots, y_{m,j}))$.  An \emph{output client} receives the output shares $z_1, \ldots, z_s$  and runs a recovery algorithm $\Rec$ to obtain $f(x_1, \ldots, x_m) = \Rec(z_1, \ldots, z_s)$. The HSS scheme $\pi$ is given by the tuple of functions $(\Share, \Eval, \Rec)$; see Definition~\ref{def:HSS} for a formal definition.  

\paragraph{Parameters of interest.}
A key parameter of interest in an HSS scheme is the \emph{download rate} (Definition~\ref{def:downloadrate}), which is the ratio of the number of bits in $f(x)$ to the number of bits in all of the output shares $z_j$.  Ideally this rate would be as close to $1$ as possible; i.e., the output client does not download much more information than it wishes to compute.

Another parameter of interest is \emph{amortization}.  As in previous work~\cite{FIKW22,BW23}, we consider HSS schemes for low-degree polynomials that amortize over $\ell$ instances of the problem. This means that we have $\ell$ batches of $m$ secrets, $x_k^{(i)}$ for $i \in [\ell]$ and $k \in [m]$, and $\ell$ polynomials $f_1, \ldots, f_\ell$.  Each of these $m\ell$ secrets is shared independently as before, but the output shares $z_j$ are allowed to depend on \emph{all} $\ell$ batches.  Then the output client is responsible for computing $f_i(x_1^{(i)}, \ldots, x_m^{(i)})$ for all $i \in [\ell]$.

\paragraph{Trade-offs between download rate and amortization.}
\cite{FIKW22,BW23} previously studied the optimal download rate possible for linear HSS schemes\footnote{A \emph{linear} HSS scheme is a scheme where both $\Share$ and $\Rec$ are linear over some field $\F$.  Note that $\Eval$ need not be linear.}, and also studied what amount of amortization is necessary to obtain this optimal rate.  
 In this work, we show that by backing off from the optimal rate by a small amount, one can get asymptotic improvements in the amortization parameter. 

In more detail, \cite{FIKW22} showed that for $t$-private, $s$-server linear HSS schemes for $m$-variate degree-$d$ polynomials, the best download rate possible is $1-dt/s$.  They achieved this download rate with schemes that had amortization $\ell = \Omega(s \log(s))$.  In follow-up work, \cite{BW23} showed that in fact amortization $\ell = \Omega(s \log(s))$ was \emph{necessary} to achieve the optimal download rate of $1 - dt/s$.  Their result followed from a characterization optimal-rate linear HSS schemes in terms of a coding theoretic notion they introduce, termed \emph{optimal labelweight codes}.

In our work, informally, we show that by backing off from the optimal rate by a small amount, we are able to get asymptotic improvements in the amortization required; in some cases we need only $\ell = O(s)$.  We obtain this by generalizing the characterization from \cite{BW23} to \emph{all} HSS schemes, not just optimal rate ones.  
We discuss our main results below in more detail.

\subsection{Main Results}\label{sec: main results}

For all of our results, we consider \emph{CNF sharing}~\cite{ISN89} (see Definition~\ref{def:CNF}).  It is known that CNF sharing is universal for linear secret sharing schemes, in that $t$-CNF shares can be locally converted to shares of \emph{any} linear $t$-private secret sharing scheme~\cite{CDI05}. 

\paragraph{Main Contributions.}

\begin{enumerate}
    \item \textbf{A complete characterization of the $\Rec$ functions for \textit{all} linear HSS schemes for $\POLYdmF$}.  As mentioned above, \cite{BW23} gives a characterization of \emph{optimal}-download-rate linear HSS schemes in terms of codes with good labelweight.  In our work, we extend that characterization to \emph{all} linear HSS schemes. 
    
    Our characterization is constructive, and in particular it gives an efficient algorithm to convert any code with good labelweight into a linear HSS scheme, and vice-versa. 

    \item \textbf{Improved amortization without much loss in rate.}  The work \cite{BW23} showed that to achieve optimal rate, it is necessary to have amortization $\ell = \Omega(s \log s)$.  Leveraging our characterization from Item 1, we give efficient constructions of linear HSS schemes that achieve \emph{near}-optimal download rate while requiring amortization parameter only $\ell = O(s)$.  We compute the parameters of our constructions in practical parameter regimes and show that our schemes achieve a near order-of-magnitude savings in amortization parameter, even for reasonable values of $s,d,m$.
\end{enumerate}

We describe our results in greater detail below.

\paragraph{(1) Characterization of arbitrary-rate linear HSS schemes}

Theorem \ref{thm:maininformal} below is a characterization of \textit{all} linear HSS schemes for $\POLYdmF$. In particular, our characterization extends that of \cite{BW23}, which only characterized optimal-rate linear HSS schemes. We show that the $\Rec$ algorithms for such schemes (with CNF sharing) are {equivalent} to a class of linear codes with sufficiently good \emph{labelweight}, a generalization of Hamming distance that was introduced by \cite{BW23}. 

 \begin{definition}[Labelweight]\label{def:LW}
     Let $\cC \subseteq \F^n$ be a linear code of dimension $\ell$.  Let $\mathcal{L}:[n] \to [s]$ be any surjective function, which we refer to as a \emph{labeling} function.  The \emph{labelweight} of $\mathbf{c} \in \cC$ is the number of distinct labels that the support of $\mathbf{c}$ touches:
     \[ \Delta_{\cL}(\mathbf{c}) = |\{ \cL(i) \,: \, i \in [n], c_i \neq 0 \}|.\]
     The \emph{labelweight} of $\cC$ is the minimum labelweight of any nonzero codeword:
     \[ \Delta_{\cL}(\cC) = \min_{c \in \cC\setminus \{0\}} \Delta_{\cL}(\mathbf{c}).\]
 \end{definition}
 In particular, if $s = n$ and $\mathcal{L}(j) = j$ for all $j \in [n]$, then $\Delta_\cL(\cC)$ is just the minimum Hamming distance of $\cC$.  Thus, the labelweight of a code generalizes the standard notion of distance.

Our main characterization theorem is the following.
\begin{theorem}[Linear HSS schemes are equivalent to labelweight codes. (Informal, see Theorem \ref{thm: main equiv})]\label{thm:maininformal}
Let $\pi = (\Share, \Eval, \Rec)$ be a $t$-private, $s$-server linear HSS for $\POLYdmF$ with download rate $R$ and amortization parameter $\ell$.  Let $G \in \mathbb{F}^{\ell \times (\ell/R)}$ be the matrix that represents $\Rec$ (see Observation~\ref{obs: linear reconstruction is matrix multiplication}). Then there is some labeling function $\cL$ so that $G$ is the generator matrix for a code $\cC$ of dimension $\ell$, with rate $R$ and with $\Delta_{\cL}(\cC) \geq dt + 1$.

Conversely, suppose that there is a labeling function $\cL:[n] \to [s]$ and a linear code $\cC \subseteq \F^n$ of dimension $\ell$ with rate $R$ and $\Delta_{\cL}(\cC) \geq dt + 1$.  Then any generator matrix $G$ of $\cC$ describes a linear reconstruction algorithm $\Rec$ for an $s$-server $t$-private linear HSS for $\POLYdmF$ that has download rate $R$ and amortization parameter $\ell$.
\end{theorem}

We remark that the converse direction is constructive: given the description of such a code $\cC$, the proof (see Theorem~\ref{thm:lblwtImpliesHSS}) gives an efficient construction of the $\Eval$ function as well as the $\Rec$ function.

\paragraph{(2) Achieving practical trade-off between download rate and amortization parameter.}\label{par: main results pt 2}

Using the complete characterization of all linear HSS schemes, we construct linear HSS schemes that achieve near-optimal download rate at amortization parameters that are strictly linear in in the number of servers $s$. While the construction of \cite{BW23} uses Reed-Solomon codes as a starting point to construct optimal rate linear HSS schemes, we use two well-studied families of algebraic geometry (AG) codes---Hermitian codes and Goppa codes.

Through the lens of Theorem \ref{thm:maininformal}, finding low amortization, high download rate linear HSS schemes is equivalent to constructing low dimension, high rate labelweight codes; that is, we wish for code dimension $\ell$ to be minimized and for code rate $\ell/n$ to be maximized. Observe that for any code dimension $\ell$, code rate is maximized when $n$ is minimized. Since $n \geq s$, we see that this occurs when setting $n = s$. In this setting, it immediately follows that the surjective labeling function $\mathcal{L}:[n]\to [s]$ can be assumed to be the identity map without loss of generality.

Therefore, our constructions use the trivial labeling scheme where $n=s$ and $\mathcal{L}:[n]\to[s], x \mapsto x$ is the identity function. Such labelweight codes, though straightforward, yield linear HSS schemes with attractive parameters for realistic server counts. Furthermore, their intuitive construction underscores the fundamental relationship between classical codes and linear HSS. We loosely summarize these results in Table \ref{fig: intro parameter comparison}; see Theorems \ref{thm: hss from hermitian code params}, \ref{thm: hss from goppa params} for additional details and formal statements.

\begin{table}[ht]
\centering
\begin{tabular}{l|c|c|c}
\hline
 & \cite{FIKW22},\cite{BW23} & Hermitian-based HSS & Goppa-based HSS \\ \hline
Download Rate & $1-dt/s$ & $ 1 - dt/s - O(s^{-1/3})$ & $1 - (dt/s) \cdot O(\log(dt))$ \\ 
Amortization & $(s-dt)\log(s)$ & $s-dt - O(s^{2/3})$ & $s - dt \cdot O(\log(dt))$ \\ \hline
\end{tabular}
\caption{Comparing our AG-based constructions to \cite{FIKW22}, \cite{BW23}}
\label{fig: intro parameter comparison}
\end{table}

Furthermore, for completeness we show in Appendix \ref{apx: linear hss from random codes} that, perhaps surprisingly, a random coding approach does not lead to amortization savings over \cite{FIKW22}, \cite{BW23}, even backing off from the optimal rate. More precisely, linear HSS schemes instantiated from random labelweight codes result in HSS schemes with amortization parameter at least $\Omega(s \log(s))$. This motivates additional study of linear codes with algebraic structure as a basis for linear HSS with attractive parameters.

\subsection{Technical Overview}\label{sec:tech}

In this section, we give a high-level overview of the techniques underpinning Theorem~\ref{thm:maininformal}, which states that any $t$-private, $s$-server linear HSS scheme for $\POLYdmF$ is equivalent to a labelweight code with minimum labelweight $\geq dt+1$.  

To give some intuition for the connection, we recount the simplest non-trivial case of the forward direction, which was proven by \cite{BW23}. We consider \emph{HSS for concatenation} \cite{FIKW22}: $\ell$ secrets $\mathbf{x} = \left(x^{(1)}, \ldots, x^{(\ell)} \right) \in \F^\ell$ are shared independently among $s$ servers who in turn communicate them to an output client. The objective is for the output client to download as little information as possible - in particular, significantly less than the naive solution of simply downloading $t+1$ shares of each secret.

Let $\mathbf{z}\in \mathbb{F}^n$ be the $n$-tuple of $\F$-symbols downloaded by the output client; since the output client instantiates a linear reconstruction algorithm $\Rec$, there exists $G \in \F^{\ell \times n}$ such that $G \mathbf{z} = \mathbf{x}$. Define a labeling function $L:[n] \to [s]$ satisfying the property that for all $i \in [n]$, $L(i) = r \in [s]$ if and only if $\mathbf{z}_i$ was downloaded from server $r$. 

\begin{claim}
    The rows of $G$ generate a linear code $\mathcal{C}$ satisfying $\Delta_\mathcal{L}(\mathcal{C}) \geq t+1$. 
\end{claim}

\begin{proof}
    Suppose towards a contradiction that for some non-zero $\mathbf{m} \in \F^\ell$, $\Delta_\mathcal{L}(\mathbf{m}G) \leq t$. Then $\mathbf{m}G \mathbf{x} = \mathbf{c} \mathbf{x}$ (for some non-zero $\mathbf{c}\in \F^n$) is a linear combination of secrets recoverable by a set of $t$ servers, contradicting the $t$-privacy they were originally secret shared with. 
\end{proof}

In this example, the function evaluated on the secret shares is identity; the generalization of this example requires consideration of more general functions, but the fundamental principle is similar.

The converse requires showing that any labelweight code $\cC$ implies a linear HSS scheme. In the setting of optimal download rate, \cite{BW23} leveraged the specific properties of optimal labelweight codes in order to prove this; their work relied on the fact that optimal labelweight codes are highly structured. In particular, \cite{BW23} showed that, in the optimal download rate setting:

\begin{itemize}
    \item[(i)] the output client must download an equal number of symbols from each server; and

    \item[(ii)] up to elementary row operations, the matrix parameterizing the output client's linear $\Rec$ algorithm is a rectangular array of invertible matrices, with the property that any square sub-array  is itself invertible.
\end{itemize}

These strong symmetry properties are key to the equivalence result of \cite{BW23}.  However, in our setting, where we do not assume that the optimal download rate is attained,  \textit{neither of the aforementioned properties hold}. Our result thus requires proving additional properties of labelweight codes; in particular, Lemma \ref{lem:updated mds-like} shows that the generator matrices of linear codes with labelweight $\geq dt+1$ have the property that any sufficiently large sub-matrix attains full row rank. This property allows us to extend the techniques introduced by \cite{BW23} and show that the specification of a linear HSS scheme can be computed from the generator matrix of a linear code with sufficient labelweight.

\subsection{Related Work}
Though linear HSS schemes are implicit in classical protocols for secure multi-party computation and private information retrieval~\cite{C:Benaloh86a,STOC:BenGolWig88,STOC:ChaCreDam88,EC:CraDamMau00,BeaverF90,BeaverFKR90,Chor:1998:PIR:293347.293350}, the systematic study of HSS was introduced by \cite{BGILT18}. Most HSS schemes reply on cryptographic hardness assumptions~\cite{BoyleGI15,DHRW16,BGI16a,BGI16,FazioGJS17,BGILT18,BoyleKS19,BCGIKS19,CouteauM21,OrlandiSY21,RoyS21,DIJL23}.

In contrast, the HSS schemes presented in this work are \emph{information-theoretically secure}.  The information-theoretic setting was explored in~\cite{BGILT18} and was further studied in \cite{FIKW22} and \cite{BW23}; these latter two works are  the closest to our work, and we discuss them more below.

The work of \cite{FIKW22} focused on the download rate of information-theoretically secure HSS schemes (both linear and non-linear) and proved a tight impossibility result regarding the highest download rate achievable by linear HSS schemes. They paired this with an explicit construction that showed a large amortization parameter is a sufficient condition for a linear HSS scheme to achieve optimal download rate.

\begin{theorem}[\cite{FIKW22}] \label{thm:FIKWInformal intro}
Let $s, d, t \in \mathbb{Z}^+$ such that $s>dt$. Let $\pi$ be a $t$-private, $s$-server linear HSS scheme for $\POLYdmF$. Then $\mathsf{DownloadRate}(\pi) \leq 1 - dt/s$.

Furthermore, for all integers $j \geq \log_{|\F|}(s)$, there exists a $t$-private, $s$-server linear HSS $\pi$ satisfying $\mathsf{DownloadRate}(\pi) = 1 - dt/s$ with amortization parameter $\ell = j(s-dt)$.
\end{theorem}

The work \cite{BW23} focused on linear schemes with \emph{optimal} download rate, meeting the bound showed in \cite{FIKW22}.  They proved that in fact a large amortization parameter is necessary for a linear HSS scheme to achieve optimal download rate.

\begin{theorem}[\cite{BW23}]\label{thm: bw23 amort informal}
    There exists a $t$-private, $s$-server linear HSS scheme for $\text{POLY}_{d,m}(\mathbb{F})$ with download rate $(s-dt)/s$ and amortization parameter $\ell$ only if $\ell = j(s-dt)$ for some $j \in \mathbb{Z}+$ satisfying
    \begin{equation*}
        j \geq \lceil \max\inset{ \log_q(s -dt +1), \log_q(dt +1)} \rceil.
    \end{equation*}
\end{theorem}

The key technique in their proof was showing that optimal-rate linear HSS is in fact equivalent to optimal labelweight linear codes; more precisely, they showed the following theorem.

\begin{theorem}[\cite{BW23}]\label{thm: bw23 equiv informal}
    There exists a $t$-private, $s$-server linear HSS scheme for $\text{POLY}_{d,m}(\mathbb{F})$ with download rate $(s-dt)/s$ and amortization parameter $\ell$ if and only if there exists a linear code $\mathcal{C} \subseteq \mathbb{F}^n$ with information rate $(s-dt)/s$ and dimension $\ell$; and surjection $L:[n] \to [s]$ such that $\Delta_L(\mathcal{C})\geq dt+1$. 
\end{theorem}


Our work extends the characterization of \cite{BW23} to \textit{all} linear HSS schemes with arbitrary download rate; in particular, we show that arbitrary-rate linear HSS schemes are equivalent to a broader class of labelweight codes than those considered by \cite{BW23}. Though our proof syntactically resembles that of Theorem \ref{thm: bw23 equiv informal} from~\cite{BW23}, as discussed in Section~\ref{sec:tech}, we need to overcome additional technical difficulties introduced by the lack of strong symmetries in the more general arbitrary-rate setting. Furthermore, we present explicit constructions that approach the optimal download rate from Theorem~\ref{thm:FIKW rate LB}, while asymptotically improving the amortization parameter $\ell$.


\subsection{Organization}

In Section~\ref{sec: prelim}, we set notation and record a few formal definitions that we will need. In Section~\ref{sec: equiv}, we show that linear HSS schemes (with arbitrary download rate) are equivalent to codes with sufficient labelweight: Lemma~\ref{lem:HSS_implies_lblwt} establishes that HSS schemes imply codes with sufficient labelweight, and Theorem~\ref{thm:lblwtImpliesHSS} constructively establishes the converse. 

In Section~\ref{sec: hermitian}, we derive labelweight codes from Hermitian codes and construct the corresponding linear HSS scheme by Theorem \ref{thm:lblwtImpliesHSS}. We formally state its parameters in Theorem \ref{thm: hss from hermitian code params} and compare its performance against constructions from \cite{FIKW22} and \cite{BW23}. In Section \ref{sec: goppa}, we do the same but use Goppa codes as the basis for a family of labelweight codes.

%% file: preliminaries.tex
We begin by setting notation and the basic definitions that we will need throughout the paper.  (We note that these definitions and notation closely follow that of \cite{FIKW22} and \cite{BW23}).

\textbf{Notation.} For $n \in \mathbb{Z}^+$, we denote by $[n]$ the set $\lbrace 1, 2, \ldots, n \rbrace$. We use bold symbols (e.g., $\mathbf{x}$) to denote vectors. For an object $w$ in some domain $\mathcal{W}$, we use $\|w\| = \log_2 (\abs{\mathcal{W}})$ to denote the number of bits used to represent $w$.

\subsection{Homomorphic Secret Sharing}

We consider homomorphic secret sharing (HSS) schemes with $m$ inputs and $s$ servers; each input is shared independently. We denote by $\mathcal{F} = \lbrace f: \mathcal{X}^m \to \mathcal{O} \rbrace$ the class of functions we wish to compute, where $\mathcal{X}$ and $\mathcal{O}$ are input and output domains, respectively. 

\begin{definition}[HSS]\label{def:HSS}

Given a collection of $s$ servers and a function class $\mathcal{F}= \lbrace f: \mathcal{X}^m \to \mathcal{O} \rbrace$, consider a tuple $\pi = (\Share, \Eval, \Rec)$, where $\Share: \mathcal{X} \times \mathcal{R} \to \mathcal{Y}^s$, $\Eval: \mathcal{F} \times [s] \times \mathcal{Y} \to \mathcal{Z}^\ast$, and $\Rec: \mathcal{Z}^\ast \to \mathcal{O}$ are as follows\footnote{By $\mathcal{Z}^\ast$, we mean a vector of some number of symbols from $\mathcal{Z}$.}:

\begin{itemize}
    \item[$\bullet$] $\Share(x_i,r_i)$: For $i \in [m]$, $\Share$ takes as input a secret $x_i \in \mathcal{X}$ and randomness $r_i \in \mathcal{R}$; it outputs $s$ 
 shares $\left( y_{i,j} : j \in [s] \right) \in \mathcal{Y}^s$. We refer to the $y_{i,j}$ as \emph{input shares}; server $j$ holds shares $(y_{i,j} 
 : i \in [m])$.

    \item[$\bullet$] $\Eval\left(f, j, \left( y_{1,j}, y_{2,j}, \ldots, y_{m,j} \right) \right)$: Given $f \in \mathcal{F}$, server index $j \in [s]$, and server $j$'s input shares $\left( y_{1,j}, y_{2,j}, \ldots, y_{m,j} \right)$, $\Eval$ outputs $z_j \in \mathcal{Z}^{n_j}$, for some $n_j \in \mathbb{Z}$. We refer to the $z_j$ as \emph{output shares}.

    \item[$\bullet$] $\Rec(z_1, \ldots, z_s)$: Given output shares $z_1, \ldots, z_s$, $\Rec$ computes $f(x_1, \ldots, x_m) \in \mathcal{O}$. 

\end{itemize}

\noindent
We say that $\pi = (\Share, \Eval, \Rec)$ is a $s$-server HSS scheme for $\mathcal{F}$ if the following requirements hold:

\begin{itemize}
    \item[$\bullet$] \textbf{Correctness:} For any $m$ inputs $x_1, \ldots, x_m \in \mathcal{X}$ and $f \in \mathcal{F}$, 

    \begin{equation*}
        \Pr_{\mathbf{r} \in \mathcal{R}^m}\left[ \Rec(z_1, \ldots, z_s) = f(x_1, \ldots, x_m) : 
            \begin{aligned} 
                &\forall i \in [m], \; \left( y_{i,1}, \ldots, y_{i,s} \right) \leftarrow \Share(x_i,r_i)\\
                &\forall j \in [s], \; z_j \leftarrow \Eval\left( f, j, (y_{1,j}, \ldots, y_{m,j}) \right)
            \end{aligned}
         \right] = 1
    \end{equation*}
Note that the random seeds $r_1, \ldots, r_m$ are independent.

    \item[$\bullet$] \textbf{Security:} Fix $i \in [m]$; we say that $\pi$ is \emph{$t$-private} if for every $T\subseteq [s]$ with $|T| \leq t$ and $x_i, x'_i \in \mathcal{X}$, $\Share(x_i)|_T$ has the same distribution as $\Share(x'_i)|_T$, over the randomness $\mathbf{r} \in \mathcal{R}^m$ used in $\Share$.

\end{itemize}

\end{definition}

\begin{remark}
    We remark that in the definition of HSS, the reconstruction algorithm $\Rec$ does \emph{not} need to know the identity of the function $f$ being computed, while the $\Eval$ function does.  In some contexts it makes sense to consider an HSS scheme for $\mathcal{F} = \{f\}$, in which case $f$ is fixed and known to all.  Our results in this work apply for general collections $\mathcal{F}$ of low-degree, multivariate polynomials, and in particular cover both situations.
\end{remark}

We focus on \emph{linear} HSS schemes, where both $\Share$ and $\Rec$ are $\F$-linear over some finite field $\F$; note that $\Eval$ need not be linear.  
\begin{definition}[Linear HSS]
    Let $\mathbb{F}$ be a finite field.
    \begin{itemize}
        \item[$\bullet$] We say that an $s$-server HSS $\pi = (\Share,\Eval,\Rec)$ has \emph{linear reconstruction} if:
        \begin{itemize}
            \item  $\mathcal{Z} = \F$, so each output share $z_i \in \F^{n_i}$ is a vector over $\F$;
            \item $\mathcal{O} = \F^o$ is a vector space over $\F$; and
            \item $\Rec: \F^{\sum_i n_i} \to \F^o$ is $\F$-linear.
        \end{itemize}

        \item[$\bullet$] We say that $\pi$ has \emph{linear sharing} if $\mathcal{X}$, $\mathcal{R}$, and $\mathcal{Y}$ are all $\F$-vector spaces, and $\Share$ is $\F$-linear.

        \item[$\bullet$] We say that $\pi$ is \emph{linear} if it has both linear reconstruction and linear sharing. Note there is no requirement for $\Eval$ to be $\mathbb{F}$-linear.
    \end{itemize}
\end{definition}

The assumption of linearity implies that the function $\Rec$ can be represented by a matrix, as per the following observation that was also used by \cite{BW23}.

\begin{observation}[\cite{BW23}]\label{obs: linear reconstruction is matrix multiplication}
    Let $\ell, t, s, d, m, n$ be integers. Let $\pi = (\Share, \Eval, \Rec)$ be a $t$-private, $s$-server HSS for some function class $\mathcal{F} \subseteq \POLYdmF^\ell$ with linear reconstruction $\Rec: \F^n \to \F^\ell$.

    Then there exists a matrix $G_\pi \in \mathbb{F}^{\ell \times n}$ so that, for all $\mathbf{f} \in \mathcal{F}$ and for all secrets $\mathbf{x} \in (\F^m)^\ell$, there exists some $\mathbf{z} \in \mathbb{F}^n$ such that
    \begin{equation*}
        \Rec(\mathbf{z}) = G_\pi \mathbf{z} = \mathbf{f}(\mathbf{x})= \left[
             \mathbf{f}_1( \mathbf{x}^{(1)}),\; \mathbf{f}_2( \mathbf{x}^{(2)}), \ldots,\; \mathbf{f}_\ell( \mathbf{x}^{(\ell)}) 
        \right]^\mathsf{T}.
    \end{equation*}
\end{observation}
For a linear HSS $\pi$, we call $G_\pi$ as in the observation above the \emph{reconstruction matrix} corresponding to $\Rec$. It was shown in \cite{BW23} that any such reconstruction matrix must be full rank.

\begin{lemma}[\cite{BW23}]\label{lem:Gpi full rank}
    Let $t,s,d,m,\ell$ be positive integers so that $m \geq d$ and $n \geq \ell$, and let $\pi$ be a $t$-private $s$-server linear HSS for some $\mathcal{F} \subseteq \POLYdmF$, so that $\mathcal{F}$ contains an element $(f_1, \ldots, f_\ell)$ where for each $i \in [\ell]$, $f_i$ is non-constant.  Then $G_\pi \in \F^{\ell \times n}$ has rank $\ell$.
\end{lemma}

Finally, we formally define the  \textit{download rate} of an HSS scheme.

\begin{definition}[Download cost, dowload rate]\label{def:downloadrate}
    Let $s, t$ be integers and let $\mathcal{F}$ be a class of functions with input space $\mathcal{X}^m$ and output space $\mathcal{O}$. Let $\pi$ be an $s$-server $t$-private HSS for $\mathcal{F}$. Let $z_i \in \mathcal{Z}^{n_i}$ for $i \in [s]$ denote the output shares.

    \begin{itemize}
        \item[$\bullet$] The \emph{download cost} of $\pi$ is given by

        \begin{equation*}
            \textsf{DownloadCost}(\pi):=\sum_{i \in [s]} \|{z_i}\|,
        \end{equation*}
        where we recall that $\|{z_i}\| = n_i \log_2|\mathcal{Z}|$ denotes the number of bits used to represent $z_i$. 
        \item[$\bullet$] The \emph{download rate} of $\pi$ is given by 
        \begin{equation*}
            \textsf{DownloadRate}(\pi):=\frac{\log_2| \mathcal{O}|}{\textsf{DownloadCost}(\pi)}.
        \end{equation*}
        
    \end{itemize}
\end{definition}

Thus, the download rate is a number between $0$ and $1$, and we would like it to be as close to $1$ as possible.

\subsection{Polynomial Function Classes}

Throughout, we will be interested in classes of functions $\mathcal{F}$ comprised of low-degree polynomials.
\begin{definition}
    Let $m > 0$ be an integer and $\mathbb{F}$ be a finite field. We define
    \begin{equation*}
        \POLYdmF := \lbrace f \in \mathbb{F}[X_1, \ldots, X_m] : \deg(f) \leq d \rbrace
    \end{equation*}
    to be the class of all $m$-variate polynomials of degree at most $d$, with coefficients in $\mathbb{F}$. 
\end{definition}
We are primarily interested in \emph{amortizing} the HSS computation over $\ell$ instances of $\POLYdmF$, as discussed in the Introduction.  We can capture this as part of Definition~\ref{def:HSS} by taking the function class $\mathcal{F}$ to be (a subset of) $\POLYdmF^\ell$ for some $\ell \in \mathbb{Z}^+$.   Note that this corresponds to the amortized setting discussed in the Introduction.

\if
As hinted at above, our infeasibility results hold not just for $\POLYdmF^\ell$ but also for any $\mathcal{F} \subseteq \POLYdmF$ that contains monomials with at least $d$ different variables.
\fi

\begin{definition}\label{def:nontriv}
    Let $\mathcal{F} \subseteq \POLYdmF^\ell$.  We say that $\mathcal{F}$ is \emph{non-trivial} if there exists some $\mathbf{f} = (f_1, \ldots, f_\ell) \in \mathcal{F}$ so that for all $i \in [\ell]$, $f_i$ contains a monomial with at least $d$ distinct variables.
\end{definition}

\if
(We note that our \emph{feasibility} results are stated in terms of $\mathcal{F} = \POLYdmF^\ell$; trivially these extend to any subset $\mathcal{F}\subseteq \POLYdmF^\ell$).
\fi

The work \cite{FIKW22} showed that any linear HSS scheme for $\POLYdmF^\ell$ (for any $\ell$) can have download rate at most $(s-dt)/s$:
We recall the following theorem from \cite{FIKW22}.
\begin{theorem}[\cite{FIKW22}]\label{thm:FIKW rate LB}
    Let $t,s,d,m,\ell$ be positive integers so that $m \geq d$. Let $\mathbb{F}$ be any finite field and $\pi$ be a $t$-private $s$-server linear HSS scheme for $\POLYdmF^\ell$. Then $dt< s$, and $\textsf{DownloadRate}(\pi) \leq (s - dt)/s$.
\end{theorem}

\subsection{CNF Sharing}
The main $\Share$ function that we consider in this work is \emph{CNF sharing}~\cite{ISN89}.

\begin{definition}[$t$-private CNF sharing] \label{def:CNF}
Let $\F$ be a finite field.
The $t$-private, $s$-server CNF secret-sharing scheme over $\F$ is a function $\Share: \F \times \F^{\binom{s}{t} - 1} \to \left(\F^{\binom{s-1}{t}}\right)^s$ that shares a secret $x \in \F$ as $s$ shares $y_j \in \F^{\binom{s-1}{t}}$, using $\binom{s}{t} - 1$ random field elements, as follows.

Let $x \in \F$, and let $\mathbf{r} \in \F^{\binom{s}{t} - 1}$ be a uniformly random vector.  Using $\mathbf{r}$, choose $y_T \in \F$ for each set $T \subseteq [s]$ of size $t$, as follows: The $y_T$ are uniformly random subject to the equation
\[ x = \sum_{T \subseteq [s]: |T| = t} y_T.\]
Then for all $j \in [s]$, define
\[ \Share(x,\mathbf{r})_j = ( y_T \,:\, j \not\in T ) \in \F^{\binom{s-1}{t} }.\]
\end{definition}

We observe that CNF-sharing is indeed $t$-private.  Any $t+1$ servers between them hold all of the shares $y_T$, and thus can reconstruct $x = \sum_T y_T$.  In contrast, any $t$ of the servers (say given by some set $S \subseteq [s]$) are missing the share $y_S$, and thus cannot learn anything about $x$.

The main reason we focus on CNF sharing is that it is \emph{universal} for linear secret sharing schemes:

\begin{theorem}[\cite{CDI05}]\label{thm:CNFuniversal}
    Suppose that $x \in \F$ is $t$-CNF-shared among $s$ servers, so that server $j$ holds $y_j \in \F^{\binom{s-1}{t}}$, and let $\Share'$ be any other linear secret-sharing scheme for $s$ servers that is (at least) $t$-private.  Then the shares $y_j$ are locally convertible into shares of $\Share'$.  That, is there are functions $\phi_1, \ldots, \phi_s$ so that
    $(\phi_1(y_1), \ldots, \phi_s(y_s))$
    has the same distribution as
    $\Share'(x, \mathbf{r})$
    for a uniformly random vector $\mathbf{r}$.
\end{theorem}

\if
In particular, we prove several results of the form ``no linear HSS with CNF sharing can do better than 
\_\_\_\_.'' Because of Theorem~\ref{thm:CNFuniversal}, these results imply that ``no linear HSS with \emph{any} linear sharing scheme can do better than \_\_\_\_.''
\fi

\subsection{Linear Codes}

Throughout, we will be working with \emph{linear codes} $\cC \subset \F^n$, which are just subspaces of $\F^n$.  For a linear code $\cC \subseteq \F^n$ of dimension $\ell$, a matrix $G \in \F^{\ell \times n}$ is a \emph{generator matrix} for $\cC$ if $\cC = \mathrm{rowSpan}(G)$. Note that generator matrices are not unique. The \emph{rate} of a linear code $\cC \subset \F^n$ of dimension $\ell$ is defined as
\[ \textsf{Rate}(\cC) := \frac{\ell}{n}. \]

%% file: equivalence.tex
In this section we show that linear HSS schemes for low-degree multivariate polynomials are equivalent to linear codes with sufficient labelweight. Concretely, we have the following theorem, which formalizes the statement of Theorem \ref{thm:maininformal}.

\begin{theorem}\label{thm: main equiv}
    Let $\ell, t, s, d, m, n$ be integers, with $m \geq d$, $\ell \leq n$. There exists a $t$-private, $s$-server $\mathbb{F}$-linear HSS $\pi = \left( \Share, \Eval, \Rec \right)$ for any non-trivial $\mathcal{F} \subseteq \text{POLY}_{d,m}(\mathbb{F})^\ell$, with download rate $\textsf{DownloadRate}(\pi)=\ell/n$, if and only if there exists a linear code $\cC \subseteq \F^n$ with rate $\textsf{DownloadRate}(\pi)$ and a labeling $\cL:[n] \to [s]$ so that $\Delta_\cL(\cC) \geq dt + 1$.
\end{theorem}

The work of \cite{BW23} proved this equivalence for only the optimal-rate setting and left the equivalence in an arbitrary-rate setting as an open question. Theorem \ref{thm: main equiv}  settles this question and shows that linear HSS and linear codes of sufficient labelweight are indeed equivalent in \textit{all} parameter regimes. The proof of Theorem~\ref{thm: main equiv} follows from Lemma~\ref{lem:HSS_implies_lblwt} (for the forward direction) and Theorem~\ref{thm:lblwtImpliesHSS} (for the converse) below.

We begin with the forward direction.
\begin{lemma}[Follows from the analysis of~\cite{BW23}]\label{lem:HSS_implies_lblwt}
    Let $\ell, t, s, d, m, n$ be integers, with $m \geq d$, $\ell \leq n$. Suppose there exists a $t$-private, $s$-server $\mathbb{F}$-linear HSS $\pi = \left( \Share, \Eval, \Rec \right)$ for any non-trivial (see Definition~\ref{def:nontriv}) $\mathcal{F} \subseteq \text{POLY}_{d,m}(\mathbb{F})^\ell$, with download rate $\textsf{DownloadRate}(\pi)=\ell/n$. Then there exists a linear code $\cC \subseteq \F^n$ with rate $\textsf{DownloadRate}(\pi)$ and a labeling $\cL:[n] \to [s]$ so that $\Delta_\cL(\cC) \geq dt + 1$.
\end{lemma}

Though the statement of Lemma \ref{lem:HSS_implies_lblwt} is more general than its optimal-rate counterpart in \cite{BW23}, the proof is analogous; a careful reading of Lemma 12 in \cite{BW23} shows that this forward direction does not leverage any of the strong symmetries of optimal rate linear HSS. Thus, we refer the reader to \cite{BW23} for a proof.

Thus, in the rest of this section we focus on the converse, which does deviate from the analysis of \cite{BW23}, as we cannot leverage the same strong symmetries that they did, as discussed in Section~\ref{sec:tech}. We first formally state the converse.

\begin{theorem}\label{thm:lblwtImpliesHSS}
    Let $\ell, t, s, d, m, n$ be integers, with $m \geq d$. Suppose that there exists a linear code $\cC \subseteq \F^n$ with dimension $\ell$ and rate $\ell/n$. Suppose there exists a labeling $\cL: [n] \to [s]$ so that $\Delta_\cL(\cC) \geq dt + 1$. Then there exists a $t$-private, $s$-server linear HSS $\pi = (\Share, \Eval, \Rec)$ for $\POLYdmF^\ell$ with download rate $\ell/n$ and amortization parameter $\ell$.
\end{theorem}

The main ingredient in proving this direction without the strong symmetries of optimal rate linear HSS is the following lemma, which neatly generalizes the results of Lemma 13 and Corollaries 14, 15 of \cite{BW23}.

\begin{lemma}\label{lem:updated mds-like}
    Let $\mathcal{C}$ be a length $n$, dimension $\ell$ linear code over $\mathbb{F}_q$ with generator matrix $G \in \mathbb{F}_q^{\ell \times n}$. Let $\mathcal{L}:[n] \to [s]$ be a surjective labeling such that $\Delta_\mathcal{L}(\mathcal{C}) \geq dt+1$.
    
    For $\Lambda \subseteq [s]$, let $G(\Lambda)$ denote the restriction of $G$ to the columns $r \in [n]$ so that $\cL(r) \in \Lambda$. Then for any $|\Lambda| \geq s-dt$, $G(\Lambda)$ has full row rank.
\end{lemma}
\begin{proof}
    Let $\Lambda = \Lambda' \cup \Lambda''$ where $\Lambda' \cap \Lambda'' = \varnothing$ and $|\Lambda'| = s-dt$. If $G(\Lambda')$ achieves full row rank, then so does $G(\Lambda)$, since adding columns to a matrix does not induce linear independence among its rows. Hence, it suffices to consider only $|\Lambda| = s-dt$.

    Up to a permutation of columns, $G$ can be written as $G = \left[ \; G(\Lambda) \; \vert \; G([s] \setminus \Lambda) \right]$. Let $w$ denote the number of columns in $G(\Lambda)$.
    
    Assume towards a contradiction that there exists some $\mathbf{v} \in \mathbb{F}_q^\ell$ such that $\mathbf{v} G(\Lambda) = 0^w$. Then 
    \begin{equation*}
        \mathbf{v}G = \left[ \mathbf{v}G(\Lambda) \; | \; \mathbf{v}G([s]\setminus\Lambda ) \right] = \left[ 0^w \; | \; \mathbf{v}G([s]\setminus\Lambda ) \right].
    \end{equation*}
    Since $|[s] \setminus \Lambda| = dt$, it follows that
    \begin{equation*}
        \Delta_\mathcal{L}(\mathbf{v}G ) = \Delta_\mathcal{L} ( \mathbf{v}G([s]\setminus\Lambda )) \leq dt
    \end{equation*}
    which contradicts $\Delta_\mathcal{L}(\mathcal{C}) \geq dt+1$.
\end{proof}

Let $G_\pi$ be a reconstruction matrix. At a high level, Lemma \ref{lem:updated mds-like} says that any sufficiently large submatrix of $G_\pi$, obtained by only considering columns labeled with a sufficiently large subset $\Lambda \subseteq [s]$, must be full-rank. \cite{BW23} proved that such a property held for $G_\pi$ in the optimal-rate regime that relied heavily on the fact that, in optimal-rate linear HSS, the output client downloads an equal number of output symbols from each server. This is equivalent to requiring that the sets $\mathcal{L}^{-1}(y) := \lbrace x \in [n] : \mathcal{L}(x) = y \rbrace$ be the same size for all $y \in [s]$. The proof of Lemma \ref{lem:updated mds-like} shows that, perhaps surprisingly, sufficiently large submatrices of $G_\pi$ still achieve full-rank even when the output client is allowed to download arbitrary numbers of output symbols from each server.

The remainder of the proof of Theorem \ref{thm:lblwtImpliesHSS} proceeds in a familiar syntax to that of \cite{BW23}; we give a compact presentation here for completeness, abridging sections that are similar to \cite{BW23} while leaving enough detail to contextualize the novel application of Lemma \ref{lem:updated mds-like} in Claim \ref{claim:fullrank}.

\begin{proof}[Proof of Theorem \ref{thm:lblwtImpliesHSS}]
    
    Let $G \in \F^{\ell \times n}$ be any generator matrix for $\cC$.
    We construct $\pi$ by its functions $\Share$, $\Eval$, and $\Rec$.  For $\Share$, we will use $t$-CNF sharing.  We define $\Rec$ using the generator matrix $G$ (as in Observation \ref{obs: linear reconstruction is matrix multiplication}).
    
    
    It remains only to define $\Eval$. Since reconstruction is linear, we may assume without loss of generality that the function $\mathbf{f} = (f_1, \ldots, f_\ell)$ that the HSS scheme is tasked to compute has $f_j(x_1, \ldots, x_m) = \prod_{i=1}^d x_i$ for all $j \in [\ell]$.
    
    Let 
    $\mathcal{T} = \lbrace T \subseteq [s] : |T| = t \rbrace$  be the set of size-$t$ subsets of $[s]$.  Let $\mathbf{x} = (\mathbf{x}^{(1)}, \ldots, \mathbf{x}^{(\ell)}) \in (\F^m)^\ell$ denote the secrets to be shared.  For $T \in \mathcal{T}$, $r \in [\ell]$, and $j \in [m]$, let $y^{(r)}_{j, T}$ denote the CNF shares of $x^{(r)}_j$, so 
    \[ x^{(r)}_j = \sum_T y^{(r)}_{j,T}.\]  Thus, for each $r \in [\ell]$,
    \begin{equation}\label{eq:want}
    f_i(\mathbf{x}^{(r)}) = \sum_{\mathbf{T} \in \mathcal{T}^d} \prod_{k=1}^d y^{(r)}_{k,T_k}. 
    \end{equation}
    Let $\mathbf{y}_j$ denote the set of CNF shares that server $j$ holds: $\mathbf{y}_j = (y_{k,T}^{(i)} : k \in [d], T \in \mathcal{T}, i \in [\ell], j \not\in T)$.  We will treat $\mathbf{y}_j$ as tuples of formal variables. We next define the following classes of monomials, in the variables $y_{j,T}^{(i)}$.  Let
\begin{equation*}
    \mathcal{M} = \left\lbrace y_{1,T_1}^{(i)} y_{2, T_2}^{(i)} \cdots y_{d, T_d}^{(i)} \; : \; \mathbf{T} \in \mathcal{T}^d, i \in [\ell] \right\rbrace.
\end{equation*}
Given a server $j \in [s]$, let $\mathcal{M}_j$ be the subset of $\mathcal{M}$ locally computable by server $j$:
\begin{equation*}
    \mathcal{M}_j  = \left\lbrace y_{1,T_1}^{(i)} y_{2, T_2}^{(i)} \cdots y_{d, T_d}^{(i)} \in \mathcal{M}  \; : \; \mathbf{T} \in \mathcal{T}^d, i \in [\ell], j \not\in \bigcup_{k \in [d]} T_k \right\rbrace.
\end{equation*}

$\Eval(\mathbf{f}, j, \mathbf{y}_j)$ determines server $j$'s output shares and is defined as $n_j = |\cL^{-1}(j)|$ polynomials of degree $d$ in the monomials $\mathcal{M}_j$.  Define a vector of variables $\mathbf{e} \in \F^{ \sum_{r \in [n]} |\mathcal{M}_{\cL(r)}|}$, indexed by pairs $(r, \chi)$ for $\chi \in \cM_{\cL(r)}$.  The vector $\mathbf{e}$ will encode the function $\Eval$ as follows.  For each $r \in [n]$,  define $z_r = z_r(\mathbf{y}_{\cL(r)})$ to be the polynomial in the variables $\mathbf{y}_{\cL(r)}$ given by
\begin{equation}\label{eq:zr}
z_r(\mathbf{y}_{\cL(r)})  := \sum_{\chi \in \cM_{\cL(r)}} \mathbf{e}_{r, \chi} \cdot \chi(\mathbf{y}_{\cL(r)}). 
\end{equation}
Then for each server $j \in [s]$ we define $\Eval$ by
\[ \Eval( \mathbf{f}, j, \mathbf{y}_j ) = ( z_r(\mathbf{y}_j) : r \in \cL^{-1}(j) ) \in \F^{n_j}. \]
To solve for the coefficients in $\mathbf{e}$ defining $\Eval$, define the matrix $S \in \F^{\ell |\cM| \times \sum_{r \in [n]}|\cM_{\cL(r)}|}$:
\begin{itemize}
    \item The rows of $S$ are indexed by pairs $(i, \fm) \in [\ell] \times \cM$.
    \item The columns of $S$ are indexed by pairs $(r, \chi)$ for $r \in [s]$ and $\chi \in \cM_r$.
    \item The entry of $S$ indexed by $(i,\fm)$ and $(r, \chi)$ is given by:
    \[ S[(i,\fm), (r, \chi)] = \begin{cases} G[i,r] & \fm = \chi \\ 0 & \text{else} \end{cases}.\]
\end{itemize}
Define a vector $\mathbf{g} \in \F^{\ell|\cM|}$ so that the coordinates of $\mathbf{g}$ are indexed by pairs $(i,\fm) \in [\ell] \times \cM$, so that
\[ \mathbf{g}[(i,\fm)] = \begin{cases} 1 & \psi_i(\fm) \\ 
0 & \text{else} \end{cases}\]
where
\[ \psi_i(\fm) = \begin{cases} 1 & \text{ $\fm$ is of the form $\prod_{k=1}^d y_{k,T_k}^{(i)}$ for some $\mathbf{T} \in \mathcal{T}^d$ } \\ 0 & \text{else} \end{cases}\]
Notice that \eqref{eq:want} implies that for $i \in [\ell]$,
\begin{equation}\label{eq:want2} f_i(\mathbf{x}^{(i)}) = \sum_{\mathbf{T} \in \mathcal{T}^d} \prod_{k=1}^d y_{k,T_k}^{(i)} = \sum_{\fm \in \mathcal{M}} \psi_i(\fm) \cdot \fm(\mathbf{y}).\end{equation}

The following claim is shown in \cite{BW23}.

\begin{claim}[\cite{BW23}, Claim 17]\label{claim:correct}
    Suppose that $S \cdot \mathbf{e} = \mathbf{g}$.  Then the HSS scheme $\pi = (\Share, \Eval, \Rec)$, where $\Share$ and $\Rec$ are as above, and $\Eval$ is defined by $\mathbf{e}$ as above, satisfies the \emph{Correctness} property in Definition~\ref{def:HSS}.
\end{claim}

Given Claim~\ref{claim:correct}, it remains only to show there exists $\mathbf{e}$ so that $S \cdot \mathbf{e} = \mathbf{g}$; this follows from the fact that $S$ is full-rank.

\begin{claim}\label{claim:fullrank}
Let $S$ be as above.  Then $S$ has full row rank.
\end{claim}
\begin{proof}
    For $\fm \in \cM$, consider the block $S^{(\fm)}$ of $S$ restricted to the rows $\{ (i, \fm) : i \in [\ell] \}$ and the columns $\{ (r, \fm) : r \in [n] \}$.  Thus, $S^{(\fm)} \in \F^{\ell \times n}$, and we may index the rows of $S^{(\fm)}$ by $i \in [\ell]$ and the columns by $r \in [n]$.  By the definition of $S$, we have
    \[S^{(\fm)}[i,r] = G[i,r] \cdot \mathbf{1}[\fm \in \cM_{\cL(r)}]. \]
     For $\fm \in \cM$, let $\Lambda_\fm = \inset{ \lambda \in [s] : \fm \in \cM_\lambda }.$   Then 
    \begin{equation}\label{eq:samerank} \mathrm{rank}(S^{(\fm)}) = \mathrm{rank}(G(\Lambda_{\fm})),
    \end{equation}
    as we can make $G(\Lambda_{\fm})$ out of $S^{(\fm)}$ by dropping all-zero columns.
    We claim that the rank in \eqref{eq:samerank} is exactly $\ell$.  To see this, first
    note that $|\Lambda_{\fm}| \geq s - dt$.  Indeed, if 
    $ \fm = \prod_{k=1}^d y^{(i)}_{k, T_k} $ for some $\mathbf{T} \in \mathcal{T}^d$, then the only $\lambda \in [s]$ so that $\fm \not\in \cM_\lambda$ are those so that $\lambda \in \bigcup_{k=1}^d T_k$, and there are at most $dt$ such $\lambda$.  Let $\Lambda_{\fm}' \subseteq \Lambda_{\fm}$ be an arbitrary subset of size $s - dt$.  Then Lemma~\ref{lem:updated mds-like} implies that $G(\Lambda_\fm)$ has full row rank.
    Equation~\ref{eq:samerank} then implies that $S^{(\fm)}$ also has full row rank, namely that it has rank $\ell$.

    Next, observe that, after rearranging rows and columns appropriately, $S$ is a block-diagonal matrix with blocks $S^{(\fm)}$ on the diagonal, for each $\fm \in \cM$.  Indeed, any entry $S[(i,\fm), (r, \chi)]$ where $\fm \neq \chi$ is zero by construction, and any entry where $\fm = \chi$ is included in the block $S^{(\fm)}$.  Thus, we conclude that $S$ too has full row-rank, as desired.
\end{proof}

Claim~\ref{claim:correct} says that any $\mathbf{e}$ so that $S\cdot\mathbf{e} = \mathbf{g}$ corresponds to a correct $\Eval$ function (with $\Share$ and $\Rec$ as given above), and Claim~\ref{claim:fullrank} implies that we can find such an $\mathbf{e}$ efficiently.  Thus we can efficiently find a description of $\pi = (\Share, \Eval, \Rec)$.  It remains to verify that $\textsf{DownloadRate}(\pi) \geq \ell /n$, which follows from construction, as the download rate of $\pi$ is equal to the rate of $\cC$, which is by definition $\ell/n$.\end{proof}

%% file: hermitian.tex
Linear HSS schemes presented by \cite{FIKW22}, \cite{BW23} both achieve optimal download rate $1-dt/s$ but require large amortization parameters to do so. \cite{FIKW22} showed it was sufficient to take amortization parameter $\ell = (s-dt) \log(s) = O(s \log(s))$, and \cite{BW23} proved that such an amortization parameter is in fact necessary in many parameter regimes, and is off by at most 1 otherwise. 

It is a natural to ask whether linear HSS schemes that achieve a better trade-off between download rate and amortization parameter exist. \textit{Specifically, can we make minor concessions to download rate and save substantially on the amortization needed?}

Through the lens of Theorem \ref{thm:lblwtImpliesHSS}, this is equivalent to asking whether there exists a labelweight code with minimum labelweight $\geq dt+1$ that achieves good rate at low dimension. A natural first attempt at an existential result would be via random coding; specifically, building a linear HSS scheme by starting with a random linear code and following the construction of Theorem \ref{thm:lblwtImpliesHSS}. Unfortunately (and perhaps surprisingly!), we show in Appendix \ref{apx: linear hss from random codes} that this results in strictly worse parameters than \cite{FIKW22}, \cite{BW23}.

In the following sections we take a different approach. We derive our labelweight codes straightforwardly from well-studied algebraic geometric codes: we set the number of servers $s$ equal to the block length $n$ of the codes and label each coordinate by the identity function $\mathcal{L}:[n] \to [s], x \mapsto x$. In this setting, labelweight is equivalent to Hamming weight. We justify this straightforward approach in subsection 2 of Section \ref{sec: main results}.

This section constructs a family of linear HSS schemes from Hermitian codes; notably, such schemes achieve asymptotically optimal download rate while requiring an amortization parameter that is only linear in $s$. 

\begin{theorem}\label{thm: hss from hermitian code params}
    Let $\ell, t, s, d, m$ be positive integers and $q$ a prime power satisfying $m \geq d$, $s-dt > 0$, and $s= q^3$. Then there exists an explicit $t$-private, $s$-server HSS $\pi = \left( \Share, \Eval, \Rec \right)$ for any non-trivial $\mathcal{F} \subseteq \mathrm{POLY}_{d,m}(\mathbb{F}_{q^2})$ with
    \begin{equation*}
        \textsf{DownloadRate}(\pi) = 1 - \frac{dt}{s} - \frac{s^{1/3}+1}{2 s^{2/3}}
    \end{equation*}
    and amortization parameter
    \begin{equation*}
        \ell = s-dt - \frac{s^{2/3}-s^{1/3}}{2}.
    \end{equation*}
\end{theorem}

We note that the above download rate is off of the optimal $1-dt/s$ by only a $O(s^{-1/3})$ term; \emph{it converges asymptotically to the optimal rate $1- dt/s$}. Furthermore, it achieves this near-optimal download rate while requiring amortization only linear in $s$. We place these parameters in the context of \cite{FIKW22}, \cite{BW23} in Figure \ref{fig: comparing hermitian hss to fikw}.

\begin{table}[ht]
\centering
\begin{tabular}{l|c|c}
\hline
 & \cite{FIKW22},\cite{BW23} & Theorem \ref{thm: hss from hermitian code params} \\ \hline
Download Rate & $1-dt/s$ & $ 1 - dt/s - O(s^{-1/3})$ \\ 
Amortization & $(s-dt)\log_{q^2}(s)$ & $s-dt - O(s^{2/3})$ \\ \hline
\end{tabular}
\caption{Comparison of Theorem \ref{thm: hss from hermitian code params} to \cite{FIKW22}, \cite{BW23}.  When $q = O(1)$ and $s = \omega(1)$, the download rate in Theorem~\ref{thm: hss from hermitian code params} approaches the optimal rate; while the amortization is asmyptotically better.}
\label{fig: comparing hermitian hss to fikw}
\end{table}

We can compare these download rates (and the amortization parameters required to achieve them) for up to 1,000 servers $s \in (dt, 1000]$ in Table \ref{table: hermitian numerical comparison}; we visualize this data in Figure~\ref{fig: comparing hermitian hss to fikw}. 

The key takeaway from these numerical illustrations of Theorem \ref{thm: hss from hermitian code params} is that \emph{even in non-asymptotic parameter regimes, small concessions in rate result in notable savings in amortization}.

\begin{table}[ht]
\centering
\begin{tabular}{c|cc|cc|cc}
\hline
\multirow{2}{*}{\# Servers} & \multicolumn{2}{c|}{\cite{FIKW22}, \cite{BW23}} & \multicolumn{2}{c|}{Theorem \ref{thm: hss from hermitian code params}} & \multicolumn{2}{c}{\% Difference} \\
 & DL Rate & Amort. & DL Rate & Amort. & DL Rate & Amort. \\
\hline
50 & 0.92 & 69 & 0.75 & 42 & -18\% & -39\% \\
100 & 0.96 & 145 & 0.83 & 88 & -13\% & -39\% \\
200 & 0.98 & 294 & 0.88 & 182 & -10\% & -38\% \\
300 & 0.98 & 444 & 0.90 & 277 & -8\% & -38\% \\
400 & 0.99 & 594 & 0.91 & 373 & -7\% & -37\% \\
500 & 0.99 & 744 & 0.92 & 469 & -7\% & -37\% \\
1000 & 0.99 & 1494 & 0.94 & 951 & -5\% & -36\% \\
\hline
\end{tabular}
\caption{Comparison of download rates, amortization parameters from \cite{FIKW22}, \cite{BW23} and Theorem \ref{thm: hss from hermitian code params} when $d = t = 2$.}
\label{table: hermitian numerical comparison}
\end{table}

\begin{figure}[ht]
    \begin{center}
        \includegraphics[scale=0.3]{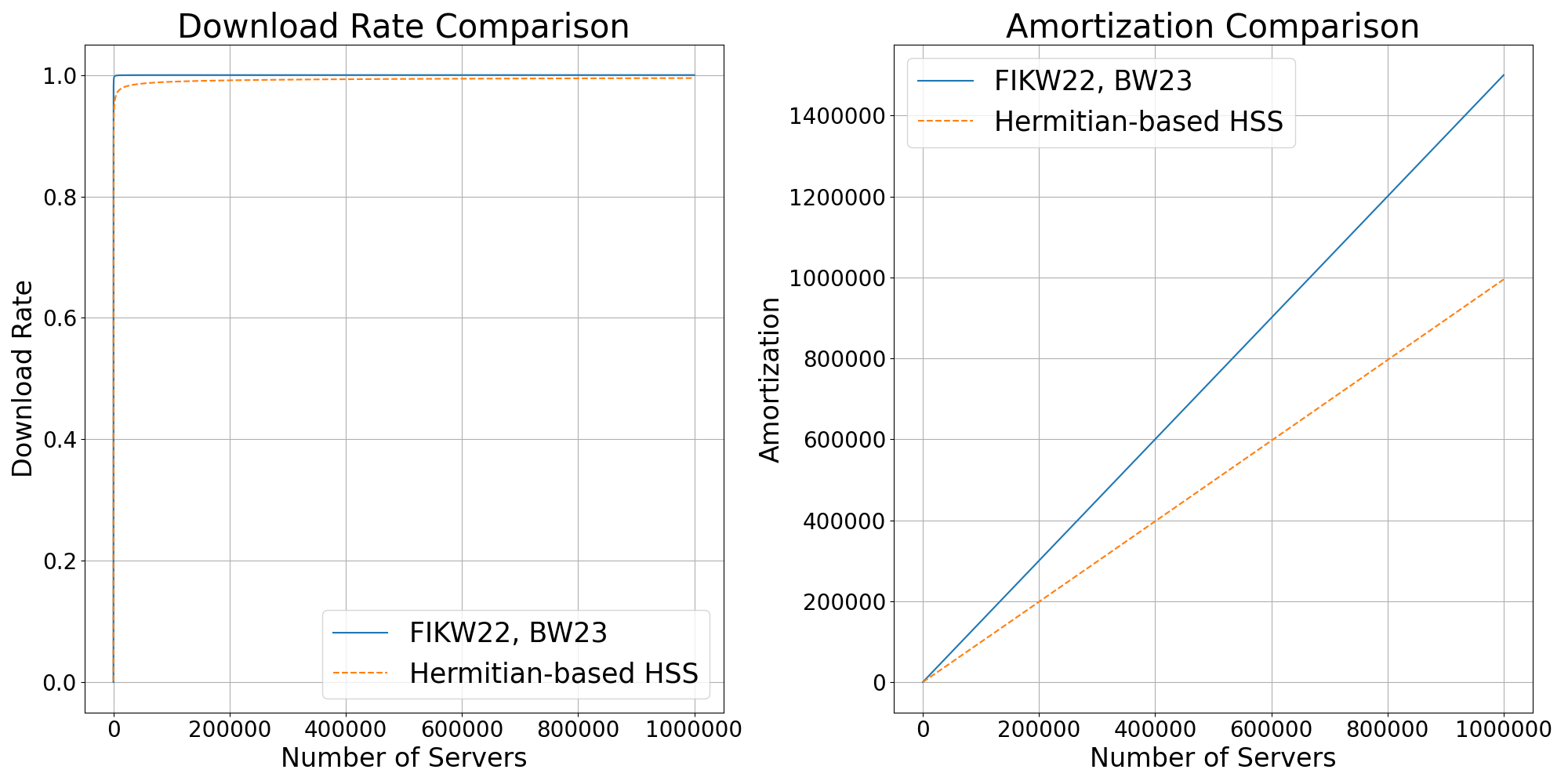}
    \caption{The left (right) plot compares the download rates (amortization parameters) of \cite{FIKW22}, \cite{BW23} with that achieved by Theorem \ref{thm: hss from hermitian code params} when $d =t = 2$. The $x$-axis denotes the number of servers and ranges from 1 to 1,000,000 to illustrate the asymptotic convergence of Theorem \ref{thm: hss from hermitian code params} to the optimal rate of \cite{FIKW22}, \cite{BW23} at a constant factor less amortization.}
    \label{fig:graphical comparison of hss vs fikw, bw}
    \end{center}
\end{figure}

\subsection{Hermitian Code Definition, Parameters}

The construction proceeds by building an optimal labelweight code from Hermitian codes before applying the construction of Theorem \ref{thm:lblwtImpliesHSS} to derive the specification of a linear HSS scheme. We begin by recalling the definition and key properties of Hermitian codes. We defer a full treatment of this well-studied family of algebraic geometry codes to \cite{Stich08}, \cite{hirschfeld2008algebraic}.

\begin{definition}[Hermitian Curve \cite{Stich08}]
The (affine) Hermitian Curve is given by the planar curve
\begin{equation*}
    g(x,y) = y^q + y - x^{q+1}.
\end{equation*}
\end{definition}

\begin{definition}[Hermitian Code \cite{hirschfeld2008algebraic}]
    Let $k \in \mathbb{Z}^+$ and let $M \subseteq \mathbb{F}_{q^2}[x,y]$ denote the set of all bi-variate polynomials $f(x,y)$ with total degree $\deg(f) < k$. Denote by
    \begin{equation*}
            \mathcal{Z}:= \left( (x,y) \in \mathbb{F}_{q^2}^2 \; : \; g(x,y) = 0 \right)
    \end{equation*}
    the affine rational points of $g$, and fix any arbitrary ordering of its elements. Then the $k$-dimensional Hermitian code $\mathcal{H}$ is given by the set of codewords
    \begin{equation*}
        \mathcal{H}:=\left\lbrace \ev_{\mathcal{Z}}(f) \; : \; f \in M \right\rbrace
    \end{equation*}
    where $\ev_\mathcal{Z}(f) = \left( f(x,y) : (x,y) \in \mathcal{Z} \right)$ denotes the standard evaluation map.
\end{definition}

\begin{theorem}[Hermitian Code Parameters \cite{hirschfeld2008algebraic}]\label{thm: hermitian code params}
    The $k$-dimensional Hermitian code $\mathcal{H}$ is a linear code of length $n=q^3$ and rate $k/n$ with minimum distance
    \begin{equation}\label{eqn: hermitian code dist}
        q^3 - k - \frac{q(q-1)}{2} + 1.
    \end{equation}
\end{theorem}

\subsection{Proof of Theorem \ref{thm: hss from hermitian code params}}

We first show the following lemma.

\begin{lemma}\label{lem: hermitian to labelweight conversion}
    Let $s = q^3,d,t \in \mathbb{Z}^+$ for some prime power $q$ such that $s-dt>0$. There exists a linear code $\mathcal{C} \subseteq \mathbb{F}_{q^2}^n$ and labeling function $L:[n] \to [s]$ satisfying $\Delta_L(\mathcal{C}) \geq dt+1$ with rate
    \begin{equation*}
        R = 1 - \frac{dt}{s} - \frac{s^{1/3}+1}{2 s^{2/3}}
    \end{equation*}
    and dimension
    \begin{equation}\label{eqn: hermitian hss dim}
        k = s-dt - \frac{s^{2/3}-s^{1/3}}{2}.
    \end{equation}
\end{lemma}

\begin{proof}
    Let $\mathcal{H}$ be the $k$-dimensional Hermitian code defined over alphabet $\mathbb{F}_{q^2}$. By Theorem \ref{thm: hermitian code params}, such a code has length $n = s = q^3$. 
    
    Allowing dimension $k$ to be as specified in Equation ~\ref{eqn: hermitian hss dim}, it follows from Equation ~\ref{eqn: hermitian code dist} that $\Delta(\mathcal{H})$ is given by
    \begin{equation*}
        s - \left( s-dt - \frac{s^{2/3}-s^{1/3}}{2} \right) - \frac{s^{1/3}(s^{1/3}-1)}{2} + 1 = dt + 1.
    \end{equation*}
    The rate of $\mathcal{H}$ is given by 
    \begin{equation*}
        R_\mathcal{H} = \frac{1}{s}\left( s-dt - \frac{s^{2/3}-s^{1/3}}{2} \right) = 1 - \frac{dt}{s} - \frac{s^{1/3}+1}{2 s^{2/3}}.
    \end{equation*}
    Set $\mathcal{H} = \mathcal{C}$ and $L: [s] \to [s], x \mapsto x$; it immediately follows that $R_\mathcal{H} = R_\mathcal{C}$ and $\Delta(\mathcal{H}) = \Delta_L(\mathcal{C}) = dt+1$, as desired.
\end{proof}

We are now prepared to prove Theorem \ref{thm: hss from hermitian code params} by applying Theorem \ref{thm:lblwtImpliesHSS}.

\begin{proof}[Proof of Theorem \ref{thm: hss from hermitian code params}]
    By Lemma \ref{lem: hermitian to labelweight conversion}, there exists a linear code $\mathcal{C} \subseteq \mathbb{F}_{q^2}^s$ and a labeling $L:[s] \to [s], x \mapsto x$ such that $\Delta_L(\mathcal{C}) \geq dt+1$; furthermore $\mathcal{C}$ has dimension 
    \begin{equation*}
        \ell = s-dt - \frac{s^{2/3} - s^{1/3}}{2}
    \end{equation*}
    and rate
    \begin{equation*}
        R = 1 - \frac{dt}{s} - \frac{s^{1/3}+1}{2 s^{2/3}}.
    \end{equation*}
    By Theorem \ref{thm:lblwtImpliesHSS}, the existence of such a labelweight code is equivalent to the existence of a linear HSS scheme with corresponding parameters; in particular, there exists a $t$-private, $s$-server, $\mathbb{F}_{q^2}$-linear HSS scheme $\pi = \pi(\Share, \Eval, \Rec)$ that achieves download rate $\mathsf{DownloadRate}(\pi) = R$ and amortization parameter $\ell$.
\end{proof}

%% file: goppa.tex
In this section we construct a family of linear HSS schemes from Goppa codes; unlike Theorem \ref{thm: hss from hermitian code params}, these schemes do not achieve asymptotically optimal rate. However, this family of schemes stands apart from those of Theorem \ref{thm: hss from hermitian code params} by allowing us to compute over the binary field regardless of the number of servers employed. Furthermore, such schemes achieve a near order-of-magnitude amortization savings at practical server counts. We first state the result before considering its performance in realistic parameter regimes.

\begin{theorem}\label{thm: hss from goppa params}
    Let $\ell, t, s, d, m, u$ be positive integers satisfying $m \geq d$, $s-dt> 0$, and $s=2^u$, where
    \begin{equation*}
        u > \log_2 \left( 2(dt)^2 - 4dt + 2(dt+1) \sqrt{(dt)^2 - 2dt + 2} + 3 \right).
    \end{equation*}
    
    Then there exists an explicit $t$-private, $s$-server HSS $\pi = \left( \Share, \Eval, \Rec \right)$ for some non-trivial $\mathcal{F} \subseteq \mathrm{POLY}_{d,m}(\mathbb{F}_{2})$ with
    \begin{equation*}
        \textsf{DownloadRate}(\pi) = 1 - u\frac{dt}{s} 
    \end{equation*}
    and amortization parameter $\ell = s-udt $.
\end{theorem}

Noting that $u \geq 3$ for all $d,t \in \mathbb{Z}^+$, we see that the download rate does not converge asymptotically to $1-dt/s$ as the construction of Theorem \ref{thm: hss from hermitian code params} does; however, we show in Table \ref{table: goppa vs fikw numerical comparison} that for small parameter values, Theorem \ref{thm: hss from goppa params} vastly outperforms Theorem \ref{thm: hss from hermitian code params} in terms of preserving rate and saving on amortization. In particular, compared to Theorem \ref{thm: hss from goppa params}, \emph{the construction of Theorem \ref{thm: hss from goppa params} concedes less rate while delivering an order-of-magnitude savings in amortization} in practical parameter regimes. We illustrate these results graphically in Figure \ref{fig:graphical comparison of goppa hss vs fikw, bw}.

\begin{table}[ht]
\centering
\begin{tabular}{c|cc|cc|cc}
\hline
\multirow{2}{*}{\# Servers} & \multicolumn{2}{c|}{\cite{FIKW22},\cite{BW23}} & \multicolumn{2}{c|}{Theorem \ref{thm: hss from goppa params}} & \multicolumn{2}{c}{\% Reduction} \\
 & DL Rate & Amortize & DL Rate & Amortize & DL Rate & Amortize \\
\hline
64 & 0.93 & 360 & 0.65 & 42 & -31\% & -88\% \\
128 & 0.96 & 868 & 0.82 & 106 & -15\% & -88\% \\
256 & 0.98 & 2016 & 0.91 & 234 & -7\% & -88\% \\
512 & 0.99 & 4572 & 0.96 & 490 & -3\% & -89\% \\
1024 & 0.99 & 10200 & 0.98 & 1002 & -1.8\% & -90\% \\
2048 & 0.99 & 22484 & 0.99 & 2026 & -0.9\% & -91\% \\
\hline
\end{tabular}
\caption{Comparison of Download Rates and Amortization Values with Percentage Differences between FIKW and Goppa.}
\label{table: goppa vs fikw numerical comparison}
\end{table}

\begin{figure}[ht]
    \begin{center}
        \includegraphics[scale=0.3]{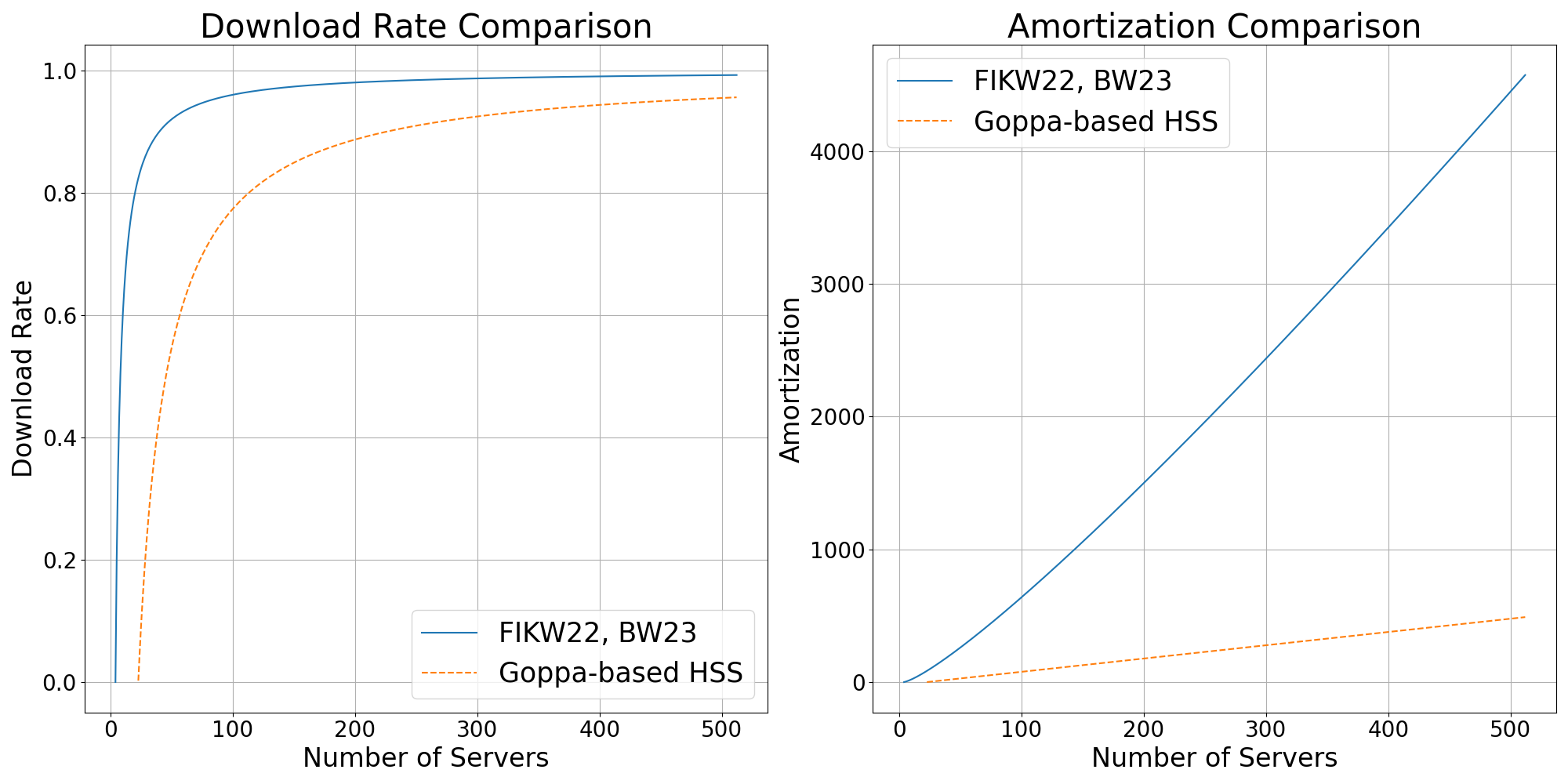}
    \caption{The left (right) plot compares the download rates (amortization parameters) of \cite{FIKW22}, \cite{BW23} with that achieved by Theorem \ref{thm: hss from goppa params} when $d =t = 2$. The $x$-axis represents the number of servers and ranges from 1 to 512. This emphasizes the super-constant amortization savings of Theorem \ref{thm: hss from goppa params} at practical parameter regimes relative to \cite{FIKW22}, \cite{BW23}, with small concessions to rate.}
    \label{fig:graphical comparison of goppa hss vs fikw, bw}
    \end{center}
\end{figure}

\subsection{Goppa Code Definition, Parameters}

The proof of Theorem \ref{thm: hss from goppa params} is constructive; it proceeds by building an optimal labelweight code from Goppa codes before applying the construction of Theorem \ref{thm:lblwtImpliesHSS} to arrive at a linear HSS scheme with the desired properties. We begin by recalling the definition and key properties of binary Goppa codes, deferring a fuller treatment to \cite{berlekamp1973goppa}, \cite{goppa1977codes}.

\begin{definition}[Goppa Polynomial \cite{berlekamp1973goppa}]
    For some $n \in \mathbb{Z}^+$, fix $V = \lbrace \alpha_1, \ldots, \alpha_n \rbrace \subseteq \mathbb{F}_{2^u},\; u \in \mathbb{Z}^+$. A Goppa polynomial is a polynomial
    \begin{equation*}
        g(x) = g_V(x) = g_0 + g_1x + \cdots + g_r x^r \in \mathbb{F}_{2^u}[x] 
    \end{equation*}
    satisfying $\deg(g) = r$ and $g(\alpha_i) \neq 0$ for all $\alpha_i \in V$.
\end{definition}

In practice, it is common to take $n = 2^u$ and set $V = \mathbb{F}_{2^u}$ for some $u \in \mathbb{Z}^+$, as it reduces the task of finding a suitable Goppa polynomial to that of finding an irreducible polynomial of degree $r$ over $\mathbb{F}_{2^u}$.

Given the definition of the Goppa polynomial above, we can define a binary Goppa code.

\begin{definition}[Goppa Codes \cite{berlekamp1973goppa}]\label{def: goppa code defn}
    Let $n, u, r \in \mathbb{Z}^+$. Fix $V = \lbrace \alpha_1, \ldots, \alpha_n \rbrace \subseteq \mathbb{F}_{2^u},\; u \in \mathbb{Z}^+$ and let $g_V \in \mathbb{F}_{2^u}[x]$ be a Goppa polynomial of degree $r$. Then the Goppa code is the set of codewords given by
    \begin{equation*}
        \Gamma_{n,u,r} = \Gamma_{n,u,r}(g,V) := \left\lbrace \mathbf{c} = (c_1, \ldots, c_n) \in \mathbb{F}_2^n \; : \; \sum_{i=1}^n \frac{c_i}{x-\alpha_i} \equiv 0 \; \mod g(x) \right\rbrace
    \end{equation*}.
\end{definition}

The parameters of Goppa codes are given by the following theorem.

\begin{theorem}[Goppa Code Parameters \cite{berlekamp1973goppa}]\label{thm: goppa code params}
    For $n, u ,r \in \mathbb{Z}^+$ let $\Gamma = \Gamma_{n,u,r}$ be a binary Goppa code as in Definition \ref{def: goppa code defn}. Then $\Gamma$ is a linear code of length $n$, dimension $k \geq n - ur$, and minimum distance $d(\Gamma) \geq r + 1$.
\end{theorem}

The parameters given by Theorem \ref{thm: goppa code params} only allow us to determine rate and minimum distance up to a lower bound, making it difficult to ascertain download rate and amortization when used to construct linear HSS schemes. Fortunately, these lower bounds are known to be sharp under additional assumptions. The following theorem gives one such instance.

\begin{theorem}[\cite{V90}]
    Fix $u ,r \in \mathbb{Z}^+$ satisfying
    \begin{equation}\label{eqn: vdv condition}
        2r - 2 < \frac{2^u - 1}{2^{u/2}}
    \end{equation}
    and let $g \in \mathbb{F}_{2^u}[x]$ be a Goppa polynomial of degree $r$ with no repeated roots. Set $V = \mathbb{F}_{2^u}$ and let $\Gamma = \Gamma_{2^u,u,r}(g,V)$ be a Goppa code as in Definition \ref{def: goppa code defn}. Then $\Gamma$ is a binary linear code with dimension precisely $k = n - ur$.
\end{theorem}

We observe that, performing the appropriate manipulations, Equation \ref{eqn: vdv condition} is satisfied for all
\begin{align}
    u& \geq \max \Big\lbrace\left\lceil\log_2 \left( 2r^2 - 4r + 2(r+1) \sqrt{r^2 - 2r + 2} + 3 \right) \right\rceil, \label{eqn: vdv condition manipulated}\\
    &\qquad\qquad \qquad \log_2 \left( 2r^2 - 4r + 2(r+1) \sqrt{r^2 - 2r + 2} + 3 \right)+ 1 \Big\rbrace. \nonumber
\end{align}

\subsection{Proof of Theorem \ref{thm: hss from goppa params}}

In this section we prove Theorem \ref{thm: hss from goppa params}. We first show the following lemma.

\begin{lemma}\label{lem: goppa to labelweight conversion}
    Let $s,d,t,u \in \mathbb{Z}^+$ satisfy $s-dt>0$ and $s = 2^u$, where
    \begin{align}\label{eqn: extension degree condition for conversion setup}
    u&= \max \Big\lbrace\left\lceil \log_2 \left( 2(dt)^2 - 4dt + 2(dt+1) \sqrt{(dt)^2 - 2dt + 2} + 3 \right) \right\rceil, \\
    &\qquad\qquad \qquad \log_2 \left( 2(dt)^2 - 4dt + 2(dt+1) \sqrt{(dt)^2 - 2dt + 2} + 3 \right)+ 1 \Big\rbrace. \nonumber
    \end{align}

    There exists a linear code $\mathcal{C} \subseteq \mathbb{F}_2^n$ and labeling function $L: [n] \to [s]$ satisfying $\Delta_L(\mathcal{C})\geq dt+1$ with rate
    \begin{equation*}
        R = 1 - \frac{udt}{s}
    \end{equation*}
    and dimension $\ell = s - udt$.
\end{lemma}

\begin{proof}
    Fix $ V = \mathbb{F}_{2^u}$ and let $g\in \mathbb{F}_{2^u}[x]$ be an irreducible polynomial of degree $r=dt$; set $n = 2^u$. Let $\Gamma = \Gamma_{n,u,r}(g,V)$ be the binary Goppa code given by Definition \ref{def: goppa code defn}. It follows from Equation \ref{eqn: extension degree condition for conversion setup} and the observation of Equation \ref{eqn: vdv condition manipulated} that $\Gamma$ has dimension $k = n - ur = s-udt$. It follows from Theorem \ref{thm: goppa code params} that $\Gamma$ has minimum distance $d(\Gamma) \geq r + 1 = dt+1$.

    Set $\mathcal{C} = \Gamma$ and define $L:[n] \to [s], x \mapsto x$ to be the identity labeling. It immediately follows that $\mathcal{C}$ has rate
    \begin{equation*}
        R = \frac{d(\Gamma)}{s} = 1- \frac{udt}{s}
    \end{equation*}
    and minimum labelweight
    \begin{equation*}
        \Delta_L(\mathcal{C}) = d(\Gamma) \geq dt+1.
    \end{equation*}
\end{proof}

It is now straightforward to prove Theorem \ref{thm: hss from goppa params} by leveraging Theorem \ref{thm:lblwtImpliesHSS}.

\begin{proof}[Proof of Theorem \ref{thm: hss from goppa params}]
By Lemma \ref{lem: goppa to labelweight conversion}, there exists a linear code $\mathcal{C} \subseteq \mathbb{F}_2^s$ and a labeling $L:[s] \to [s]$ such that $\Delta_L(\mathcal{C}) \geq dt+1$; furthermore $\mathcal{C}$ has dimension $\ell = s - udt$ and rate $R = 1- udt/s$. By Theorem \ref{thm:lblwtImpliesHSS}, the existence of such a labelweight code is equivalent to the existence of a linear HSS scheme with corresponding parameters; in particular, there exists a $t$-private, $s$-server, $\mathbb{F}_2$-linear HSS scheme $\pi = \pi(\Share, \Eval, \Rec)$ that achieves download rate 
\begin{equation*}
    \mathsf{DownloadRate}(\pi) = R = 1 - \frac{udt}{s}
\end{equation*}
and amortization parameter $\ell = s- udt$. 
\end{proof}

%% file: gv_generalization.tex
In this section we show that the natural random coding approach does not appear to yield linear HSS schemes that meaningfully outperform the $\ell = O(s \log(s))$ amortization parameter required by \cite{FIKW22}, \cite{BW23}.  Indeed, the standard argument established that random linear codes correspond to linear HSS schemes that can attain good download rates, but---like~\cite{BW23,FIKW22}---only with large amortization parameters. 

\subsection{Notation}

To justify the notion that ``random labelweight codes don't outperform Reed-Solomon codes in linear HSS amortization'', we proceed by generalizing the well-known Gilbert-Varshamov Bound to the labelweight setting. 
One standard proof of this result (see, e.g.,~\cite{ECT}) analyzes the distance of a random linear code, and we follow the same path here.  We first introduce some notation.

\begin{definition}[Labelweight Ball]
    Let $L:[n] \to [s]$ be a surjective labeling. We define the labelweight ball $B_L(r)$ of radius $0 \leq r \leq n$ to be the set
    \begin{equation*}
        B_L(r) := \left\lbrace c \in \mathbb{F}_q^n \; : \; \Delta_L(c) \leq r \right\rbrace
    \end{equation*}
    and define the volume of the labelweight ball to be $\mathrm{Vol}_L(r) = | B_L(r) |$.
\end{definition}

For the purposes of our analysis, we consider only a fixed labeling function.

\begin{assume}\label{assume: fixed balanced labeling}
    Let $n, s, w \in \mathbb{Z}^+$ such that $ n = sw$. In this section we will only consider the labeling
    \begin{equation*}
        L: [n] \to [s], \; x \mapsto \left\lceil \frac{x}{s} \right\rceil.
    \end{equation*}
\end{assume}
When paired with a code of length $n$, this balanced labeling simply labels the first $w$ coordinates with 1, the second $w$ coordinates with 2, and continues analogously until the last $w$ coordinates are labeled with $s$. Intuitively, for fixed $n, s \in \mathbb{Z}^+$, such a balanced labeling pattern maximizes labelweight in expectation over random linear codes.

Under this fixed, balanced labeling function of Assumption \ref{assume: fixed balanced labeling}, we have the following algebraic formulation of labelweight ball volume.

\begin{observation}
    Let $n, s, w,r \in \mathbb{Z}^+$ such that $n = sw$ and $0 \leq r \leq n$. Then
    \begin{equation*}
        \mathrm{Vol}_L(r) = | B_L(r) | = \sum_{i=0}^r \binom{s}{i}\left( q^w - 1 \right)^i.
    \end{equation*}
\end{observation}

Finally, we will need to define relative labelweight for labelweight codes, which is the natural analogue of relative minimum distance for linear codes.

\begin{definition}[Relative Labelweight]\label{def: relative labelweight}
    Let $\mathcal{C} \subseteq \mathbb{F}_q^n$ be a linear code and $L: [n] \to [s]$ a surjective labeling such that $\Delta_L(\mathcal{C}) = d$. We define the relative labelweight of $\mathcal{C}$ to be $\delta = d/s$.
\end{definition}

\subsection{Generalization of \texorpdfstring{$q$}{q}-ary Entropy}

In the standard proof of the Gilbert-Varshamov bound, the volume of a Hamming ball is estimated by the $q$-ary entropy function.  To generalize the proof to labelweight, we introduce the following generalization of the $q$-ary entropy function, which captures the volume of labelweight balls.

\begin{definition}[Generalized $q$-ary Entropy]
    Let $q \geq 2$, $w \geq 1$. For $x \in (0,1)$, we denote by $H_{q,w}(x)$ the generalized $q$-ary entropy function:
    \begin{equation*}
       H_{q,w}(x) = x \log_q(q^w - 1) - x \log_q(x) - (1-x) \log_q(1-x),
    \end{equation*}
    where $H_{q,w}(0), H_{q,w}(1)$ are defined as the limit of $H_{q,w}$ as $x \to 0, 1$, respectively.
\end{definition}

Note that the case where $w = 1$ is the standard $q$-ary entropy function. We notice that, when properly normalized, the generalized entropy function can be approximated linearly.

\begin{observation}\label{obs: gen entropy linear approx}
    For all $x \in [0, 1 - 1/q^w]$, $x \leq w^{-1}H_{q,w}(x) \leq x + \log_q(2^{1/w})$.
\end{observation}
\begin{proof}
    Observe that 
    \begin{align*}
        g(x):=w^{-1}H_{q,w}(x) - x &= w^{-1}\left(x \log_q(q^w - 1) - x \log_q(x) - (1-x) \log_q(1-x)\right) - x\\
        &\leq w^{-1} \left( - x \log_q(x) - (1-x) \log_q(1-x) \right).
    \end{align*}
    Since $- x \log_q(x) - (1-x) \log_q(1-x)$ is a concave function which attains its maximal value when $x = 1/2$, it follows that $w^{-1}H_{q,w}(x) - x \leq w^{-1}  \log_q(2)$ as desired. The lower bound follows from observing that $g$ is itself a concave function, since
    \begin{equation*}
        g''(x) = -\frac{1}{w \cdot x \cdot (1 - x) \cdot \ln(q)} \leq 0 \; \forall x \in [0, 1 - 1/q^w]
    \end{equation*}
    and that its values at the endpoints of the domain $[0, 1-1/q^w]$ are non-negative; i.e., $g(0), g(1-1/q^w) \geq 0$.
\end{proof}

Equipped with this definition, our goal becomes to express the volume of a given labelweight ball in terms of the generalized entropy function. To do so, we note two helpful relations; we omit the proofs, which are elementary algebraic manipulations.

\begin{observation}\label{obs: remove entropy function from exponent}
    Let $s, p \in \mathbb{R}$ such that $s,p \geq 0$. Then
    \begin{equation*}
        q^{-sH_{q,w}(p)} = (1-p)^{(1-p)s} \left( \frac{p}{q^w - 1} \right)^{ps}
    \end{equation*}
\end{observation}

\begin{observation}\label{obs: bounding ratio}
    Let $w \in \mathbb{Z}^+$ and $p \in [0,1)$ satisfy $0 \leq p \leq 1 - 1/q^w$. Then
    \begin{equation*}
        \frac{p}{(1-p)(q^w-1)} \leq 1.
    \end{equation*}
\end{observation}

We now give the volume of a labelweight ball in terms of the generalized entropy function.

\begin{lemma}\label{lem: ball volume in terms of entropy}
    Let $s,w \in \mathbb{Z}^+$ and $p \in [0,1)$ satisfy $0 \leq p \leq 1 - 1/q^w$ and $ps \in \mathbb{Z}^+$. Then
    \begin{equation*}
        \mathrm{Vol}_L(ps) \leq q^{s H_{q,w}(p)}.
    \end{equation*}
\end{lemma}

\begin{proof}
Observe that 
\begin{align*}
1 &= (p + (1-p))^s= \sum_{i=0}^{s} \binom{s}{i} p^i (1-p)^{s-i}\geq \sum_{i=0}^{ps} \binom{s}{i} p^i (1-p)^{s-i}.
\end{align*}
Multiplying through by $1 = (q^w-1)^i/(q^w-1)^i$ and applying Observation 
\ref{obs: bounding ratio} yields
\begin{equation*}
    1 \geq \sum_{i=0}^{ps} \binom{s}{i}  (q^w-1)^i (1-p)^{s} \left(\frac{p}{(1-p)(q^w-1)}\right)^{ps}.
\end{equation*}
Finally, applying Observation \ref{obs: remove entropy function from exponent} yields
\begin{align*}
    1 &\geq \sum_{i=0}^{ps} \binom{s}{i}  (q^w-1)^i q^{-sH_{q,w}(p)}\\
&= \mathrm{Vol}_L(ps) q^{-sH_{q,w}(p)}.
\end{align*}
    
\end{proof}

\subsection{Gilbert-Varshamov Bound for Random Labelweight Codes}

We are finally equipped to prove a generalization of the Gilbert-Varshamov bound for labelweight codes. This generalization will quantify the rate, and labelweight trade-off we can guarantee through random linear codes; viewed through the lens of Theorem \ref{thm:lblwtImpliesHSS}, this tells us the download rate and amortization parameters that can be guaranteed by linear HSS scheme  constructed from random linear codes. 

\begin{theorem}\label{thm: gv for labelweight}
    For $q \geq 2$, let $n, s, w \in \mathbb{Z}^+$ satisfy $n = sw$. Let $\delta \in [0, 1 - 1/q^w]$ satisfy $\delta s \in \mathbb{Z}^+$. For $\varepsilon \in [0, 1 - H_{q,w}(\delta)]$, let
    \begin{equation}\label{eqn: lgv dimension}
        k = n - sH_{q,w}(\delta) - n\varepsilon
    \end{equation}
    and let $G \in \mathbb{F}_q^{k \times n}$ be chosen uniformly at random.

    Then with probability $> 1 - q^{-\varepsilon n}$, $G$ is the generator matrix of a length $n$, dimension $k$, and relative labelweight $\geq \delta$ linear code with rate
    \begin{equation*}
        R = 1 - \frac{sH_{q,w}(\delta)}{n} - \varepsilon.
    \end{equation*}
\end{theorem}

Note that when $n = s$ and $w=1$, Theorem \ref{thm: gv for labelweight} becomes the standard Gilbert-Varshamov Bound.
Before we show the proof of Theorem \ref{thm: gv for labelweight}, we interpret its statement in terms of linear HSS parameters. 

\begin{example}\label{ex: gv example}
Let $s, d, t \in \mathbb{Z}^+$ satisfying $s-dt> 0$ parameterize a linear HSS scheme as in Definition \ref{def:HSS}. Let $s$ be as stated in Theorem \ref{thm: gv for labelweight} and set $\delta = (dt+1)/s$.

For the sake of illustration, suppose $w = \log_q(s)$ and $\varepsilon > 0$ a negligible constant. Let $\mathcal{C}$ denote the linear code with properties guaranteed by Theorem \ref{thm: gv for labelweight} and let $\pi$ denote the $t$-private, $s$-server linear HSS constructed from $\mathcal{C}$ as in Theorem \ref{thm:lblwtImpliesHSS}. Applying Observation \ref{obs: gen entropy linear approx} to Theorem \ref{thm:lblwtImpliesHSS}, $\pi$ has download rate at most
\begin{equation*}
    \mathsf{DownloadRate}(\pi) \leq 1 - \frac{dt+1}{s} - \varepsilon = 1 - \frac{dt}{s} - O(s^{-1})
\end{equation*}
with amortization parameter \textit{at least}
\begin{equation*}
    \ell \geq (1-\varepsilon)s \log_q(s) - s\log_q(2) - (dt+1)\log_q(s) = \Omega(s\log(s))
\end{equation*}
for sufficiently small $\varepsilon$. In particular, we note that such a construction has an amortization parameter (at least) on the same $\Omega(s \log(s))$ order as that of \cite{FIKW22}, \cite{BW23}, while achieving a rate comparable to that of our Hermitian code-based construction of Theorem \ref{thm: hss from hermitian code params}. We summarize this situation in Table \ref{fig: comparing hermitian hss to gv}.
\end{example}

\begin{table}[ht]
\centering
\begin{tabular}{l|c|c}
\hline
 & Thm. \ref{thm: hss from hermitian code params} (Hermitian code-based) & Ex. \ref{ex: gv example} (Random code-based) \\ \hline
Download Rate & $1 - dt/s - O(s^{-1/3})$ & $\leq 1 - dt/s - O(s^{-1})$ \\ 
Amortization & $s-dt - O(s^{2/3})$ & $\Omega(s\log(s))$ \\ \hline
\end{tabular}
\caption{Comparison of Theorem \ref{thm: hss from hermitian code params} to Example \ref{ex: gv example}.}
\label{fig: comparing hermitian hss to gv}
\end{table}

We conclude this section by proving Theorem \ref{thm: gv for labelweight}.

\begin{proof}[Proof of Theorem \ref{thm: gv for labelweight}]
Let $\mathcal{C} = \lbrace \mathbf{m}G \; : \; \mathbf{m} \in \mathbb{F}_q^k \rbrace$ be the linear code generated by $G$. It suffices to show that $\Delta_L(\mathbf{m}G) \geq d$ for all non-zero $\mathbf{m}$.

    Accordingly, let $\mathbf{m} \in \mathbb{F}_q^k$ be a uniformly random non-zero vector; then $\mathbf{m}G$ is uniformly distributed over $\mathbb{F}_q^n$. It follows from Lemma \ref{lem: ball volume in terms of entropy} that
    \begin{equation*}
        \Pr \left[ \delta_L(\mathbf{m}G) < d \right] = \frac{\mathrm{Vol}_L(d-1)}{q^n} \leq \frac{q^{s H_{q,w}(\delta)}}{q^n} = q^{-k} \; q^{-n\varepsilon}.
    \end{equation*}
    Taking the Union Bound over all $\mathbf{m} \in \mathbb{F}_q^k$ yields the observation that with probability $1 - q^{-n\varepsilon}$, $\Delta_L(\mathcal{C}) \geq d$ as desired.
\end{proof}